\documentclass[usenames,dvipsnames]{llncs}
\usepackage[utf8]{inputenc}
\usepackage[T1]{fontenc}
\usepackage{url}
\usepackage[colorlinks]{hyperref}
\usepackage{algorithm}
\usepackage[noend]{algpseudocode}
\usepackage{amsmath,amssymb}
\usepackage{graphicx}

\begin{document}

\title{Evaluation of Chebyshev polynomials on intervals and application
to root finding}

\author{%
  Viviane Ledoux\inst{1,2} \and
  Guillaume Moroz\inst{1}
}
\institute{%
  INRIA, France {\tt Firstname.Name@inria.fr} \and
  École Normale Supérieure {\tt Firstname.Name@ens.fr}.
}

\maketitle

\begin{abstract}
  In approximation theory, it is standard to approximate functions by
  polynomials expressed in the Chebyshev basis. Evaluating a polynomial $f$
  of degree n given in the Chebyshev basis can be done in $O(n)$ arithmetic
  operations using the Clenshaw algorithm. Unfortunately, the evaluation
  of $f$ on an interval $I$ using the Clenshaw algorithm with interval
  arithmetic returns an interval of width exponential in $n$. We describe
  a variant of the Clenshaw algorithm based on ball arithmetic that
  returns an interval of width quadratic in $n$ for an interval of small
  enough width. As an application, our variant of the Clenshaw algorithm
  can be used to design an efficient root finding algorithm.

\end{abstract}

\section{Introduction}

Clenshaw showed in 1955 that any polynomial given in the form
\begin{equation}
  p(x) = \sum_{i=0}^n a_i T_i(x)
\end{equation}
can be evaluated on a value $x$ with a single loop using the following
functions defined by recurrence:
\begin{equation}
  \label{eq:clenshaw}
  u_k(x) = 
  \begin{cases}
    0   & \text{if } k = n+1\\
    a_n & \text{if } k = n\\
    2xu_{k+1}(x) - u_{k+2}(x) + a_k & \text{if } 1\leq k < n\\
    x u_1(x) - u_2(x) + a_0         & \text{if } k = 0
  \end{cases}
\end{equation}
such that $p(x) = u_0(x)$.


Unfortunately, if we use Equation~\eqref{eq:clenshaw} with interval
arithmetic directly, the result can be an interval of size exponentially
larger than the input, as illustrated in Example~\ref{ex:interval}.

\begin{example}\label{ex:interval}
  Let $\varepsilon > 0$ be a positive real number, and let $x$ be the
interval $[\frac 12 - \varepsilon, \frac 12 + \varepsilon]$ of width
$2\varepsilon$. Assuming that $a_n = 1$, we can
see that $u_{n-1}$ is an interval of width $4\varepsilon$. Then by
recurrence, we observe that $u_{n-k}$ is an interval of width at least
$4\varepsilon F_k$ where ${(F_n)}_{n\in\mathbb N}$ denotes the Fibonacci
sequence, even if all $a_i=0$ for $i<n$.
\end{example}

Note that the constant below the exponent is even higher when $x$ is
closer to $1$. These numerical instabilities also appear with floating
point arithmetic near $1$ and $-1$ as analyzed in~\cite{Gcj69}.

To work around the numerical instabilities near $1$ and $-1$, Reinsch
suggested a variant of the Clenshaw algorithm~\cite{Gcj69,Ojima77}.
Let $d_n(x) = a_n$ and  $u_n(x) = a_n$, and for $k$ between $0$ and $n-1$,
define $d_k(x)$ and $u_k(x)$ by recurrence as follows:
\begin{equation}
  \label{eq:reinsch}
  \begin{cases}
    d_k(x) = 2(x-1)u_{k+1}(x) + d_{k+1}(x) +  a_k\\
    u_k(x) = d_k(x) + u_{k+1}
  \end{cases}
\end{equation}
Computing $p(x)$ with this recurrence is numerically more stable near
$1$. However, this algorithm does not solve the problem of exponential
growth illustrated in Example~\ref{ex:interval}.

Our first main result is a generalization of Equation~\ref{eq:reinsch}
for any value in the interval $[-1, 1]$. This leads to
Algorithm~\ref{alg:ballclenshaw_forward} that returns intervals with
tighter radii, as analyzed in Lemma~\ref{lem:analysis}. Our second main
result is the use of classical backward error analysis to derive
Algorithm~\ref{alg:ballclenshaw_backward} which gives an even better
radii. Then in Section~\ref{sec:solve} we use the new evaluation
algorithm to design a root solver for Chebyshev series, detailed in
Algorithm~\ref{alg:subdivision}.

\section{Evaluation of Chebyshev polynomials on intervals}
\subsection{Forward error analysis}
\label{sec:evaluate}
In this section we assume that we want to evaluate a Chebyshev
polynomial on the interval $I$. Let $a$ be the center of $I$ and $r$ be
its radius. Furthermore, let $\gamma$ and $\overline \gamma$ be the $2$
conjugate complex roots of the equation:
\begin{equation}
  \label{eq:quadratic}
  X^2 - 2aX + 1 = 0.
\end{equation}
In particular, using Vieta's formulas that relate the coefficients
to the roots of a polynomial, $\gamma$ satisfies $\gamma + \overline
\gamma = 2a$ and $\gamma\overline \gamma = 1$.

Let $z_n(x) = a_n$ and  $u_n(x) = a_n$, and for $k$ between $0$ and $n-1$,
define $z_k(x)$ and $u_k(x)$ by recurrence as follows:
\begin{equation}
  \begin{cases}
    z_k(x) = 2(x-a)u_{k+1}(x) + \gamma z_{k+1}(x) +  a_k\\
    u_k(x) = z_k(x) + \overline \gamma u_{k+1}(x)
  \end{cases}
\end{equation}

Using Equation~\eqref{eq:quadratic}, we can check that the $u_k$
satisfies the recurrence relation $u_k(x) = 2xu_{k+1}(x) - u_{k+2}(x) + a_k$,
such that $p(x) = xu_1(x) - u_2(x) + a_0$.

Let $(e_k)$ and $(f_k)$ be two sequences of positive real numbers. Let
$\mathcal B_\mathbb R(a,r)$ and $\mathcal B_\mathbb R(u_k(a), e_k)$
represent the intervals $[a-r, a+r]$ and $[u_k(a)-e_k, u_k(a)+e_k]$. Let
$\mathcal B_\mathbb C(z_k(a), f_k)$ be the complex ball of center
$z_k(a)$ and radius $f_k$.

Our goal is to compute recurrence formulas on the $e_k$ and the $f_k$ such that:
\begin{equation}
  \label{eq:inclusion}
  \begin{cases}
    \begin{aligned}
      z_k(\mathcal B_\mathbb R(a,r)) &\subset \mathcal B_\mathbb C(z_k(a), f_k)\\
      u_k(\mathcal B_\mathbb R(a,r)) &\subset \mathcal B_\mathbb R(u_k(a), e_k).
    \end{aligned}
  \end{cases}
\end{equation}

\begin{lemma}
  Let $e_n=0$ and $f_n=0$ and for $n > k \geq 1$:
  \begin{equation}
    \begin{cases}
      f_k = 2r|u_{k+1}(a)| + 2re_{k+1} + f_{k+1}\\
      e_k = \min(e_{k+1} + f_k, \frac{f_k}{\sqrt{1-a^2}}) \text{ if $|a|<1$  else } e_{k+1} + f_k
    \end{cases}
  \end{equation}
  Then, $(e_k)$ and $(f_k)$ satisfy Equation~\eqref{eq:inclusion}.
\end{lemma}
\begin{proof}[sketch]
  For the inclusion $z_k(\mathcal B_\mathbb R(a,r)) \subset \mathcal
  B_\mathbb C(z_k(a), f_k)$, note that $\gamma$ has modulus $1$, such
  that the radius of $\gamma z_{k+1}$ is the same as the radius of
  $z_{k+1}$ when using ball arithmetics. The remaining terms bounding the radius of $z_k$ follow
  from the standard rules of interval arithmetics.

  For the inclusion $u_k(\mathcal B_\mathbb R(a,r)) \subset \mathcal
  B_\mathbb C(u_k(a), e_k)$, note that the error segment on $u_k$ is included in
  the Minkowski sum of a disk of radius $f_k$ and a segment of radius
  $e_{k+1}$, denoted by $M$. If $\theta$ is the angle of the segment with the
  horizontal, we have $\cos \theta = a$. We conclude that the
  intersection of $M$ with a horizontal line is a segment of radius at most
  $\min(e_{k+1} + f_k, \frac{f_k}{\sqrt{1-a^2}})$.

\end{proof}

\begin{corollary}
  Let $\mathcal B_\mathbb R(u,e)=\verb|BallClenshawForward|((a_0, \ldots, a_n), a, r)$ be the
  result of Algorithm~\ref{alg:ballclenshaw_forward}, then
  \begin{equation*}
    p(\mathcal B_\mathbb R(a, r)) \subset \mathcal B_\mathbb R(u, e)
  \end{equation*}
\end{corollary}

Moreover, the following lemma bounds the radius of the ball returned by
Algorithm~\ref{alg:ballclenshaw_forward}.
\begin{lemma}
  \label{lem:analysis}
  Let $\mathcal B_\mathbb R(u,e)=\verb|BallClenshawForward|((a_0, \ldots, a_n), a, r)$ be the
  result of Algorithm~\ref{alg:ballclenshaw_forward}, and
  let $M$ be an upper bound on $|u_k(a)|$ for $1 \leq k \leq
  n$. Assume that $\varepsilon_k < Mr$ for $1 \leq k \leq n$, then
  \begin{equation*}
    \left\{
      \begin{array}{llll}
        e < 2Mn^2r &\text{if}\quad n < &\frac1{2\sqrt{1-a^2}}& \\
        e < 9Mn\frac{r}{\sqrt{1-a^2}}   &\text{if} &\frac1{2\sqrt{1-a^2}} \leq n < &\frac{\sqrt{1-a^2}}{2r}\\
        e < 2M\left[(1+\frac{2r}{\sqrt{1-a^2}})^n-1\right]   &\text{if} & &\frac{\sqrt{1-a^2}}{2r} < n
    \end{array}
    \right.
  \end{equation*}
\end{lemma}
\begin{proof}[sketch]
  We distinguish $2$ cases. First if $n < \frac1{2\sqrt{1-a^2}}$, we
  focus on the relation $e_k \leq e_{k+1} + f_k + Mr$, and we prove by
  descending recurrence that $e_k \leq 2M(n-k)^2r$ and $f_k \leq 2Mr(2(n-k-1)+1)$.
  
  For the case $\frac1{2\sqrt{1-a^2}} \leq n$, we use the relation $e_k
  \leq \frac{f_k}{\sqrt{1-a^2}} + Mr$, that we substitute in the
  recurrence relation defining $f_k$ to get $f_k \leq 2rM +
  \frac{2r}{\sqrt{1-a^2}} f_{k+1} + f_{k+1} + Mr\sqrt{1-a^2}$. We can
  check by recurrence that $f_k \leq
  \frac32M\sqrt{1-a^2}\left[(1+\frac{2r}{\sqrt{1-a^2}})^n-1\right]$,
  which allows us to conclude for the case $\frac{\sqrt{1-a^2}}{2r} \leq
  n$. Finally, when $\frac1{2\sqrt{1-a^2}} \leq n <
  \frac{\sqrt{1-a^2}}{2r}$, we observe that
  $(1+\frac{2r}{\sqrt{1-a^2}})^n-1 \leq n\exp(1)\frac{2r}{\sqrt{1-a^2}}$
  which leads to the bound for the last case.

\end{proof}

\begin{algorithm}
  \caption{Clenshaw evaluation algorithm, forward error}
  \label{alg:ballclenshaw_forward}
  \begin{algorithmic}
    \Function{BallClenshawForward}{$(a_0,\ldots, a_n)$, $a$, $r$}
      \State \(\triangleright\) \emph{Computation of the centers $u_k$}
      \State $u_{n+1} \gets 0$
      \State $u_n \gets a_n$
      \For{$k$ in $n-1, n-2, \ldots, 1$}
        \State $u_k \gets 2au_{k+1} - u_{k+2} + a_k$
        \State $\varepsilon_k \gets$ bound on the rounding error for $u_k$
      \EndFor
      \State $u_0 \gets a u_1 - u_2 + a_0$
      \State $\varepsilon_0 \gets$ bound on the rounding error for $u_0$
      \State
      \State \(\triangleright\) \emph{Computation of the radii $e_k$}
      \State $f_n \gets 0$
      \State $e_n \gets 0$
      \For{$k$ in $n-1, n-2, \ldots, 1$}
        \State $f_k \gets 2r|u_{k+1}| + 2re_{k+1} + f_{k+1}$
        \State $e_k \gets min(e_{k+1} + f_k, \frac{f_k}{\sqrt{1-a^2}}) + \varepsilon_k$
      \EndFor
      \State $f_0 \gets r|u_{1}| + 2re_{1} + f_{1}$
      \State $e_0 \gets min(e_1 + f_0, \frac{f_0}{\sqrt{1-a^2}}) + \varepsilon_0$
      \State \Return $\mathcal B_\mathbb R(u_0, e_0)$
    \EndFunction
  \end{algorithmic}
\end{algorithm}

\subsection{Backward error analysis}
In the literature, we can find an error analysis of the Clenshaw
algorithm \cite{Ejams68}.
The main idea is to add the errors appearing at each step of the
Clenshaw algorithm to the input coefficients. Thus the approximate
result correspond to the exact result of an approximate input. Finally,
the error bound is obtained as the evaluation of a Chebyshev polynomial.
This error analysis can be used directly
to derive an algorithm to evaluate a polynomial in the Chebyshev basis
on an interval in Algorithm~\ref{alg:ballclenshaw_backward}. 

\begin{lemma}
  Let $e_n=0$ and for $n > k \geq 1$:
  \begin{equation}
      e_k = 2r|u_{k+1}(a)| + e_{k+1}
  \end{equation}
  and $e_0 = r|u_1(a)| + e_1$.  Then $(e_k)$ satisfies 
  $u_k(\mathcal B_\mathbb R(a,r)) \subset \mathcal B_\mathbb R(u_k(a), e_k)$.
\end{lemma}
\begin{proof}[sketch]
In the case where the computations are performed without errors, D.
Elliott \cite[Equation~(4.9)]{Ejams68} showed that for $\gamma = \tilde x - x$ we have:
$$p(\tilde x) - p(x) = 2\gamma \sum_{i=0}^n u_i(\tilde x) T_i(x) - \gamma u_1(\tilde x)$$

In the case where $\tilde x = a$ and $x \in \mathcal B_{\mathbb R}(a,r)$
we have $\gamma \leq r$ and $|T(x)| \leq 1$ which implies $e_k \leq
r|u_1(a)| + \sum_{i=2}^n 2r|u_i(a)|$.
\end{proof}
\begin{corollary}
  Let $\mathcal B_\mathbb R(u,e)=\verb|BallClenshawBackward|((a_0, \ldots, a_n), a, r)$ be the
  result of Algorithm~\ref{alg:ballclenshaw_backward}, and
  let $M$ be an upper bound on $|u_k(a)|$ for $1 \leq k \leq
  n$. Assume that $\varepsilon_k < Mr$ for $1 \leq k \leq n$, then $e < 3Mnr$.
\end{corollary}

\begin{algorithm}
  \caption{Clenshaw evaluation algorithm, backward error}
  \label{alg:ballclenshaw_backward}
  \begin{algorithmic}
    \Function{BallClenshawBackward}{$(a_0,\ldots, a_n)$, $a$, $r$}
      \State \(\triangleright\) \emph{Computation of the centers $u_k$}
      \State $u_{n+1} \gets 0$
      \State $u_n \gets a_n$
      \For{$k$ in $n-1, n-2, \ldots, 1$}
        \State $u_k \gets 2au_{k+1} - u_{k+2} + a_k$
        \State $\varepsilon_k \gets$ bound on the rounding error for $u_k$
      \EndFor
      \State $u_0 \gets a u_1 - u_2 + a_0$
      \State $\varepsilon_0 \gets$ bound on the rounding error for $u_0$
      \State
      \State \(\triangleright\) \emph{Computation of the radii $e_k$}
      \State $e_n \gets 0$
      \For{$k$ in $n-1, n-2, \ldots, 1$}
      \State $e_k \gets e_{k+1} + 2r|u_{k+1}| + \varepsilon_k$
      \EndFor
      \State $e_0 \gets e_{1} + r|u_{1}| + \varepsilon_0$
      \State \Return $\mathcal B_\mathbb R(u_0, e_0)$
    \EndFunction
  \end{algorithmic}
\end{algorithm}

\section{Application to root finding}
\label{sec:solve}

For classical polynomials, numerous solvers exist in the literature, such
as those described in \cite{KRSissac16} for example. For polynomials in
the Chebyshev basis, several approaches exist that reduce the problem to
polynomial complex root finding \cite{Bjna02}, or complex eigenvalue
computations \cite{Bsiam13} among other.

In this section, we experiment a direct subdivision algorithm based on
interval evaluation, detailed in Algorithm~\ref{alg:subdivision}.
This algorithm is implemented and publicly available in the software
\verb|clenshaw| \cite{Mclenshaw19}.

\begin{algorithm}
  \caption{Subdivision algorithm for root finding}
  \label{alg:subdivision}
  \begin{algorithmic}
    \Require{
      \begin{tabular}[t]{l}
        $(a_0, \ldots, a_n)$ represents the Chebyshev polynomial approximating $f(x)$\\
        $(b_0, \ldots, b_n)$ represents the Chebyshev polynomial approximating $\frac{df}{dx}(x)$
      \end{tabular}
    }
    \Ensure{$Res$ is a list of isolating intervals for the roots of $f$ in $[-1, 1]$}
    \Function{SubdivideClenshaw}{$(a_0, \ldots, a_n)$, $(b_0,\dots, b_n)$}
      \State \(\triangleright\) \emph{Partition $[-1,1]$ in intervals where $F$ either has constant sign or is monotonous}
      \State $L \gets [ \mathcal B_\mathbb R(0, 1)]$
      \State $Partition \gets []$
      \While{$L$ is not empty}
      \State $\mathcal B_\mathbb R(a, r) \gets$ \texttt{pop} the first element of $L$
        \State $\mathcal B_\mathbb R(f, s) \gets \texttt{BallClenshaw}\left((a_0, \ldots, a_n), a, r\right)$
        \State $\mathcal B_\mathbb R(df, t) \gets \texttt{BallClenshaw}\left((b_0, \ldots, b_n), a, r\right)$
        \If{$f-s > 0$}
          \State \texttt{append} the pair $(\mathcal B_\mathbb R(a,r),
          "plus")$ to $Partition$
        \ElsIf{$f+s < 0$}
          \State \texttt{append} the pair $(\mathcal B_\mathbb R(a,r),
          "minus")$ to $Partition$
        \ElsIf{$g-s > 0$ or $g+s < 0$}
          \State \texttt{append} the pair $(\mathcal B_\mathbb R(a,r),
          "monotonous")$ to $Partition$
        \Else
          \State $\mathcal B_1, \mathcal B_2 \gets \texttt{subdivide} \mathcal B_\mathbb R (a, r)$
          \State \texttt{append} $\mathcal B_1, \mathcal B_2$ to $L$
        \EndIf
      \EndWhile
      \State
      \State \(\triangleright\) \emph{Compute the sign of $F$ at the boundaries}
      \State $\mathcal B_\mathbb R(f, s) \gets \texttt{BallClenshaw}\left((a_0, \ldots, a_n), -1, 0\right)$
      \State \texttt{append} the pair $(\mathcal B_\mathbb R(-1,0),
      \texttt{sign}(\mathcal B_\mathbb R(f, s)))$ to $Partition$
      \State $\mathcal B_\mathbb R(f, s) \gets \texttt{BallClenshaw}\left((a_0, \ldots, a_n), 1, 0\right)$
      \State \texttt{append} the pair $(\mathcal B_\mathbb R(1,0),
      \texttt{sign}(\mathcal B_\mathbb R(f, s)))$ to $Partition$
      \State
      \State \(\triangleright\) \emph{Recover the root isolating intervals}
      \State $Partition \gets$ \texttt{sort} $Partition$
      \State $Res \gets$
        \parbox[t]{0.8\linewidth}{the "monotonous" intervals of $Partition$\\
          such that the adjacent intervals have opposite signs}
      \State \Return $Res$
    \EndFunction
  \end{algorithmic}
\end{algorithm}

We applied this approach to Chebyshev polynomials
whose coefficients are independently and identically distributed with
the normal distribution with mean $0$ and variance $1$.

As illustrated in Figure~\ref{fig:experiment} our code performs
significantly better than the classical companion matrix approach. In
particular, we could solve polynomials of degree $90000$ in the Chebyshev
basis in less than $5$ seconds and polynomials of degree $5000$ in
$0.043$ seconds on a quad-core Intel(R) i7-8650U cpu at 1.9GHz. For
comparison, the standard numpy function \verb|chebroots| took more than
$65$ seconds for polynomials of degree $5000$.  Moreover, using least
square fitting on the ten last values, we observe that our approach has
an experimental complexity closer to $\Theta(n^{1.67})$, whereas the
companion matrix approach has a complexity closer to $\Theta(n^{2.39})$.
\begin{figure}
  \includegraphics[width=\linewidth]{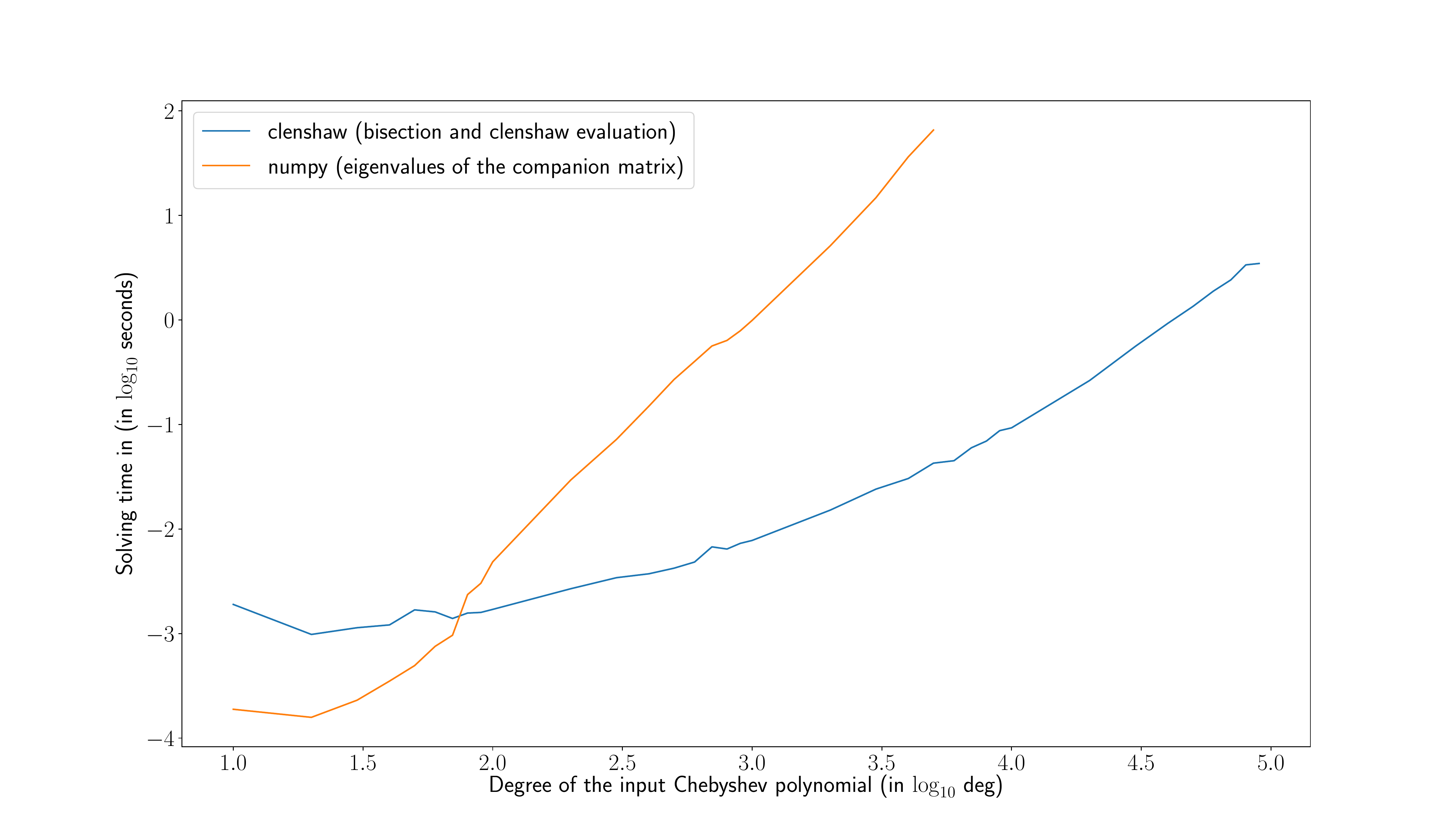}
  \caption{Time for isolating the roots of a random Chebyshev polynomial, on a
  quad-core Intel(R) i7-8650U cpu at 1.9GHz, with 16G of ram}
  \label{fig:experiment}
\end{figure}
\bibliographystyle{splncs04}
\bibliography{chebyshev}
\end{document}